\pdfoutput=1
\RequirePackage{ifpdf}
\ifpdf 
\documentclass[pdftex]{sigma}
\else
\documentclass{sigma}
\fi

\numberwithin{equation}{section}

\newtheorem{teo}{Theorem}[section]
\newtheorem{prop}[teo]{Proposition}

{\theoremstyle{definition}

\newtheorem{rmk}[teo]{Remark}
}

\begin{document}
\allowdisplaybreaks

\newcommand{\arXivNumber}{2001.08613}

\renewcommand{\PaperNumber}{052}

\FirstPageHeading

\ShortArticleName{On the Extended-Hamiltonian Structure}

\ArticleName{On the Extended-Hamiltonian Structure of Certain\\ Superintegrable Systems on Constant-Curvature\\ Riemannian and Pseudo-Riemannian Surfaces}

\Author{Claudia Maria CHANU and Giovanni RASTELLI}

\AuthorNameForHeading{C.M.~Chanu and G.~Rastelli}

\Address{Dipartimento di Matematica, Universit\`a di Torino, Torino, Italia}
\Email{\href{mailto:claudiamaria.chanu@unito.it}{claudiamaria.chanu@unito.it}, \href{mailto:giovanni.rastelli@unito.it}{giovanni.rastelli@unito.it}}

\ArticleDates{Received March 21, 2020, in final form May 20, 2020; Published online June 11, 2020}

\Abstract{We prove the integrability and superintegrability of a family of natural Hamiltonians which includes and generalises those studied in some literature, originally defined on the 2D Minkowski space. Some of the new Hamiltonians are a perfect analogy of the well-known superintegrable system on the Euclidean plane proposed by Tremblay--Turbiner--Winternitz and they are defined on Minkowski space, as well as on all other 2D manifolds of constant curvature, Riemannian or pseudo-Riemannian. We show also how the application of the coupling-constant-metamorphosis technique allows us to obtain new superintegrable Hamiltonians from the previous ones. Moreover, for the Minkowski case, we show the quantum superintegrability of the corresponding quantum Hamiltonian operator.Our results are obtained by applying the theory of extended Hamiltonian systems, which is strictly connected with the geometry of warped manifolds.}


\Classification{37J35; 70H33}

\section{Introduction}\label{section1}

In \cite{MPT} it is proved that the Hamiltonian
\begin{gather}\label{H0}
H_0=p_1p_2-\alpha q_2^{2k+1}q_1^{-2k-3},
\end{gather}
$\alpha \in \mathbb R$, is superintegrable for any $k \in \mathbb Q$. The proof is obtained by introducing an irregular bi-Hamiltonian structure, in analogy with the rational Calogero--Moser system.

A partial generalization of this Hamiltonian has been later proved to be again superintegrable in \cite{CS}. This is the Hamiltonian
\begin{gather}\label{H}
H=p_1p_2-\alpha q_2^{2k+1}q_1^{-2k-3}-\frac \beta 2 q^k_2q_1^{-k-2},
\end{gather}
with $k \in \mathbb N$, $\alpha, \beta \in \mathbb R$.
In both the previous articles, the superintegrability of the Hamiltonians is proved via algebraic methods.
In this article, we prove that $H$ is indeed superintegrable for any $k \in \mathbb Q$ and that its Hamilton--Jacobi equation is separable, implying the integrability of the system, for any $k \in \mathbb R$. Moreover, we determine a new class of Hamiltonians which gene\-ra\-lizes~$H$ and is again superintegrable. We obtain these results by showing that $H$ can be written in form of an extended Hamiltonian. We introduced extended Hamiltonian systems in~\cite{CDRPol} and in~\cite{CDRfi,CDRgen,CDRraz,TTWcdr} we generalized the idea of extended systems. An extended Hamiltonian is a~Hamiltonian $H$ on a cotangent bundle,
with $n+1$ degrees of freedom, which is obtained from a~Hamiltonian $L$ with $n$ degrees of freedom, depending on a rational parameter $k$ and admitting a non-trivial characteristic and explicitly determined first integral~$K$. When~$L$ is polynomial in the momenta, the same is for~$H$ and for~$K$. Moreover, if $L$ is a natural Hamiltonian on the cotangent bundle of a Riemannian (or pseudo-Riemannian) manifold, i.e., $L$ is a Hamiltonian of mechanical type, then its extensions (if any) are natural Hamiltonians too.
Even if in the most general setting the Hamiltonians involved in the extension procedure are not necessarily natural (see~\cite{NoNN} for examples of extensions of non-natural Hamiltonians), in this article, however, we apply the general theory to natural Hamiltonians only, because we are dealing with the natural Hamiltonian~(\ref{H}) and some generalizations of it.
The term {\it extension} is motivated by the fact that our procedure increases the degrees of freedom of the Hamiltonian system, this is the only relation with other existing extension procedures, for example the Eisenhart extension. Since the characteristic first integral is generated by the recursive application of a~particular operator, our extension procedure is essentially algebraic, but it admits a geometric characterization in the case of natural Hamiltonians.
The class of extended Hamiltonians, parametrized by $k\in \mathbb Q$, includes many (but not all, see Remark~\ref{na}) of the known superintegrable Hamiltonians admitting a polynomial in the momenta first integral of degree dependent on any chosen rational number~$k$. Examples which belong to the above class are, for instance, the isotropic and non-isotropic harmonic oscillators, the three-particle Jacobi--Calogero system, the Wolfes system, the Tremblay--Turbiner--Winternitz system, the Post--Winternitz system. In particular, this last system is studied in \cite{CDRpw} by applying the technique of coupling-constant-metamorphosis. We apply here the same technique to obtain from~$H$, and our generalizations of it, new superintegrable systems. In Section~\ref{ex} we recall the basic elements of the theory of extended systems. In Section~\ref{s3} we show that $H$ is indeed an extended Hamiltonian. In Section~\ref{section4} we apply the theories of extended Hamiltonians and coupling-constant-metamorphosis to obtain new superintegrable systems. Moreover we show that the same family of extended Hamiltonians of $H$ includes Hamiltonians on other constant-curvature two-dimensional manifolds: Euclidean planes, spheres, pseudo-spheres, de~Sitter and anti-de~Sitter spaces. These Hamiltonians can be considered as generalisations of $H$ to constant-curvature manifolds. In Section~\ref{section5}, the Laplace--Beltrami quantization of $H$ and its characteristic first integral, together with their generalisations on flat manifolds, is obtained via the shift-ladder approach proposed by S.~Kuru and J.~Negro in~\cite{KN} and adapted to extended Hamiltonians in~\cite{CRsl}.

\section{Extensions of Hamiltonian systems}\label{ex}
Let $L\big(q^i,p_i\big)$ be any Hamiltonian function with $N$ degrees of freedom, defined on the cotangent bundle of an $N$-dimensional manifold with coordinates $\big(q^i\big)$ and conjugate momenta $(p_i)$.

We say that $L$ {\em admits extensions}, if there exists $(c,c_0)\in \mathbb R^2
\setminus \{(0,0)\}$ such that there exists a non null solution $G\big(q^i,p_i\big)$ of
\begin{gather}\label{e1}
X_L^2(G)=-2(cL+c_0)G,
\end{gather}
where $X_L$ is the Hamiltonian vector field of $L$.

If $L$ admits extensions, then, for any $\gamma(u)$ solution of the ODE
\begin{gather}\label{eqgam}
\gamma'+c\gamma^2+C=0,
\end{gather}
depending on the arbitrary constant parameter~$C$, we say that any Hamiltonian $H\big(u,q^i,p_u,p_i\big)$ with $N+1$ degrees of freedom of the form
\begin{gather}\label{Hest}
H=\frac{1}{2} p_u^2-\left(\frac mn \right)^2\gamma'L+ \left(\frac mn \right)^2c_0\gamma^2
+\frac{\Omega}{\gamma^2},
\qquad (m,n)\in \mathbb{N}\setminus \{0\}, \quad \Omega\in\mathbb{R}
\end{gather}
is an {\em extension of $L$}.

Extensions of Hamiltonians where introduced in \cite{CDRfi} and studied because, when $L$ is a polynomial in the momenta Hamiltonian and in particular when~$L$ is a natural Hamiltonian, they admit a polynomial in the momenta first integral generated via a recursive algorithm;
its degree in $(p_u,p_i)$
 depends on the value of
 $m,n\in \mathbb{N}\setminus \{0\}$.
For a generic Hamiltonian $L$, given any $m,n\in \mathbb N\setminus\{0\}$,
let us consider the operator
\begin{gather}\label{Umn}
U_{m,n}=p_u+\frac m{n^2}\gamma X_L.
\end{gather}

\begin{prop}[\cite{CDRraz}]For $\Omega=0$, the Hamiltonian \eqref{Hest}
is in involution with the function
\begin{gather}\label{mn_int}
K_{m,n}=U_{m,n}^m(G_n)=\left(p_u+\frac{m}{n^2} \gamma(u) X_L\right)^m(G_n),
\end{gather}
where $G_n$ is the $n$-th term of the recursion
\begin{gather}\label{rec}
G_1=G, \qquad G_{n+1}=X_L(G) G_n+\frac{1}{n}G X_L(G_n),
\end{gather}
starting from any solution $G$ of \eqref{e1}.
\end{prop}

For $\Omega\neq 0$, the recursive construction of a first integral is more complicated: we construct the following function, depending on two strictly positive integers~$s$,~$r$
\begin{gather} \label{ee2}
\bar K_{2s,r}=\left(U_{2s,r}^2+2\Omega \gamma^{-2}\right)^s(G_r),
\end{gather}
where
the operator $U^2_{2s,r}$ is defined {according to~(\ref{Umn})} as
 \begin{gather*}
U^2_{2s,r}=\left( p_u+\frac{2s}{r^2} \gamma(u) X_L \right)^2,
 \end{gather*}
and~$G_r$ is, as in (\ref{mn_int}), the $r$-th term of the recursion~(\ref{rec}),
with $G_1=G$ solution of~(\ref{e1}). For $\Omega=0$ the functions~(\ref{ee2})
reduce to~(\ref{mn_int}) and can be computed also when the first of the indices is odd.

\begin{teo}[\cite{TTWcdr}] \label{t2} For any $\Omega\in \mathbb{R}$, the Hamiltonian~\eqref{Hest} satisfies, for $m=2s$,
\begin{gather*}
\big\{H,\bar K_{m,n}\big\}=0,
\end{gather*}
for $m=2s+1$,
\begin{gather*}
\big\{H ,\bar K_{2m,2n}\big\}=0.
\end{gather*}
\end{teo}

We call $K$ and $\bar{K}$, of the form (\ref{mn_int}) and~(\ref{ee2}) respectively, \emph{characteristic first integrals} of the corresponding extensions. It is proved in \cite{CDRfi,TTWcdr} that the characteristic first integrals~$K$ or~$\bar K$ are functionally independent from $H$, $L$, and from any first integral $I\big(p_i,q^i\big)$ of~$L$. This means that the extensions of (maximally) superintegrable Hamiltonians are (maximally) superintegrable Hamiltonians with one additional degree of freedom (see also~\cite{CDRsuext}). In particular, any extension of a one-dimensional Hamiltonian is maximally superintegrable.

The explicit expression of the characteristic first integrals is given as follows~\cite{TTWcdr}.
For $r\leq m$, we have
\begin{gather*}
U_{m,n}^r(G_n)= P_{m,n,r}G_n+D_{m,n,r}X_{L}(G_n),
\end{gather*}
with
 \begin{gather*}
G_n=\sum_{ j=0}^{[(n-1)/2]}\binom{n}{2 j+1} G^{2 j+1}(X_L G)^{n-2 j-1}(-2)^j(cL+c_0)^ j,
\\
P_{m,n,r}=\sum_{j=0}^{[r/2]}\binom{r}{2 j} \left(\frac mn \gamma \right)^{2 j}p_u^{r-2 j}(-2)^j(cL+c_0)^ j,
\\
D_{m,n,r}=\frac 1{n}\sum_{ j=0}^{[(r-1)/2]}\binom{r}{2 j+1} \left(\frac mn \gamma \right)^{2 j+1}p_u^{r-2 j-1}(-2)^j(cL+c_0)^ j, \qquad m>1,
 \end{gather*}
where $[\cdot]$ denotes the integer part and $D_{1,n,1}=\frac 1{n^2} \gamma$.

The expansion of the first integral (\ref{ee2}) is
\begin{gather*}
\bar K_{2m,n}=\sum_{j=0} ^{m}\binom{m}{j}\left(\frac {2\Omega}{\gamma^2}\right)^jU_{2m,n}^{2(m-j)}(G_n),
\end{gather*}
with $U^0_{2m,n}(G_n)=G_n$.

\begin{rmk}If $L$ is a natural Hamiltonian with a metric and a potential which are algebraic (resp.\ rational) functions of the $\big(q^i\big)$ and if $L$ admits an extension associated with a~function~$G\big(q^i,p_i\big)$ which is polynomial in the momenta and algebraic (resp.\ rational) in the~$\big(q^i\big)$, then the extension of $L$ with $\gamma=(cu)^{-1}$ has first integrals which are polynomial in the momenta and algebraic (resp.\ rational) in the~$\big(q^i\big)$. The first integrals $\bar K_{2m,n}$ are polynomial w.r.t.~$G$,~$L$,~$X_L G$,~$\gamma$ and $p_u$. Indeed, by (\ref{e1}), it is easy to check that $X_L(G_n)$ is a polynomial in~$G$,~$X_L G$,~$L$, as well as~$G_n$.
\end{rmk}

\begin{rmk}In \cite{CDRgen} it is proven that the ODE (\ref{eqgam}) defining $\gamma$ is a necessary condition in order to get a characteristic first integral of the form~(\ref{mn_int}) or~(\ref{ee2}). According to the value of~$c$ and~$C$, the explicit form of~$\gamma(u)$ is given (up to constant translations of~$u$) by
\begin{gather*}
\gamma=
\begin{cases}
-C u, & c=0, \\
\dfrac{1}{T_\kappa (cu)}=\dfrac{C_\kappa (cu) }{S_\kappa (cu)},
& c\neq 0 ,
\end{cases}
\end{gather*}
where $\kappa=C/c$ is the ratio of the constant parameters appearing in (\ref{eqgam}) and $T_\kappa$, $S_\kappa$ and $C_\kappa$ are the trigonometric tagged functions
\cite{7ter} (see also \cite{CDRraz} for a summary of their properties)
 \begin{gather*}
S_\kappa(x)= \begin{cases}
\dfrac{\sin\sqrt{\kappa}x}{\sqrt{\kappa}}, & \kappa>0, \\
x, & \kappa=0, \\
\dfrac{\sinh\sqrt{|\kappa|}x}{\sqrt{|\kappa|}} & \kappa<0
\end{cases}
\qquad
C_\kappa(x)= \begin{cases}
\cos\sqrt{\kappa}x, & \kappa>0 ,\\
1, & \kappa=0, \\
\cosh\sqrt{|\kappa|}x, & \kappa<0,
\end{cases}
\qquad
T_\kappa(x)=\frac {S_\kappa(x)}{C_\kappa(x)}.
 \end{gather*}
Therefore, we have
\begin{gather*}
\gamma'=
\begin{cases}
-C, & c=0, \\
\dfrac{-c}{S_\kappa ^2 (cu)}, & c\neq 0 .
\end{cases}
\end{gather*}
\end{rmk}

\begin{rmk}\label{RemarkA}
When $c\neq 0$ and $\kappa\neq 0$, a translation of the variable $u$ allows not so evident changes in $\gamma$ and $\gamma'$. Indeed, by using standard formulas, we can write for example, for $c\neq 0$,
\begin{gather*}
\gamma'(cu)=\frac{4\kappa c}{({\rm e}^{\sqrt{-\kappa}cu}-{\rm e}^{-\sqrt{-\kappa}cu})^2},\qquad \gamma^2(cu)= -\kappa\frac{({\rm e}^{\sqrt{-\kappa}cu}+{\rm e}^{-\sqrt{-\kappa}cu})^2}{({\rm e}^{\sqrt{-\kappa}cu}-{\rm e}^{-\sqrt{-\kappa}cu})^2}
\end{gather*}
so that, for $c=\pm 1$, $\kappa=\pm 1$ we can easily compute, in correspondence with the indicated translations of $u$, the expressions of~$\gamma'$
\begin{center}
\begin{tabular}{ | l | l |l|}
\hline
{\it Table of $\gamma'$} & $\kappa=1$ & $\kappa=-1$ \tsep{2pt}\\ \hline
 $c=1$ & $- \sin ^{-2} u$ & $-\sinh ^{-2} u$\tsep{2pt} \\ \hline
$c=-1$ & $\sin ^{-2} u$ & $\sinh ^{-2} u$ \tsep{2pt}\\ \hline
 translation: & $u\rightarrow u+\pi/2$ & $ u\rightarrow u+{\rm i} \pi/2$ \tsep{2pt} \\ \hline
 $c=1$ & $- \cos ^{-2} u$ & $\cosh ^{-2} u$ \tsep{2pt}\\ \hline
 $c=-1$ & $\cos ^{-2} u$ & $- \cosh ^{-2} u$ \tsep{2pt}\\ \hline
\end{tabular}
\end{center}
and of $\gamma ^2$
\begin{center}
\begin{tabular}{ | l | l |l|}
\hline
{\it Table of $\gamma^2$} & $\kappa=1$ & $\kappa=-1$ \tsep{2pt}\\ \hline
$c=\pm 1$ & $\tan ^{-2} u$ & $\tanh ^{-2} u$ \tsep{2pt}\\ \hline
translation: & $u\rightarrow u+\pi/2$ & $ u\rightarrow u+{\rm i} \pi/2$ \tsep{2pt}\\ \hline
$c=\pm 1$ & $\tan ^{2} u$ & $\tanh ^{2} u$ \tsep{2pt}\\ \hline
\end{tabular}
\end{center}
where the complex translation of the real variable $u$ leaves~(\ref{Hest}) real.
The table shows the remarkable fact that, for both positive and negative values of~$c$, we can always choose positive or negative values of~$\gamma'$, by selecting suitable translations and values of $\kappa$, i.e., we can obtain extensions with different signature. This is particularly important since, as pointed out in~\cite{CDRgen}, if $L$ is a natural Hamiltonian with more than one degree of freedom, then the integrability conditions of~(\ref{e1}) require that $c$ satisfies certain conditions (related with the sectional curvatures of the metric of~$L$), therefore, the sign of $c$ is fixed by $L$, while $\kappa$ remains free. Hence, without taking in account translations in $u$, the signature of the metric of the extended manifold (i.e., the sign of $\gamma '$ in (\ref{Hest})) will be constrained by $L$. Actually, that sign can be freely chosen.
\end{rmk}

In order to check if any $(N+1)$-dimensional Hamiltonian $H$ is an extension of a $N$-dimensional Hamiltonian $L$, we must check if
\begin{enumerate}\itemsep=0pt
\item[(i)] There exist canonical coordinates $(u,p_u)$ such that $H$ can be written as a \emph{warped Hamiltonian}, that is
\begin{gather}\label{Hext}
H=\frac 12 p_u^2+f(u)+\left(\frac {m}n\right)^2\alpha(u) L, \qquad m,n\in \mathbb{N}\setminus\{0\},
\end{gather}
and the Hamiltonian $L$ is independent of $(u,p_u)$. As we see in the next section, these canonical coordinates can be determined by following an invariant geometrical procedure, when $H$ is a natural Hamiltonian.
\item[(ii)]For some constants $c$ and $c_0$ not both vanishing, the equation (\ref{e1})
admits a non null solution $G$.
\item[(iii)] The functions $\alpha$ and $f$ in \eqref{Hext} can be written as in~(\ref{Hest}) for a $\gamma$ satisfying~(\ref{eqgam}). \end{enumerate}

We remark that if (\ref{e1}) has a solution for $c\neq 0$, then we may assume without loss of generality $c_0=0$, because $L$ is determined up to additive constants.

Examples of the procedure of extension applied on Hamiltonians~$L$ which are not of mecha\-ni\-cal type, defined possibly on Poisson manifolds, are given in~\cite{NoNN}. In the following, we assume that all Hamiltonian functions are natural and defined on cotangent bundles of Riemannian or pseudo-Riemannian manifolds.

\section{Determination of the extended Hamiltonian structure} \label{s3}

In order to show that the natural Hamiltonian (\ref{H}) is an extended Hamiltonian, we have first to check if it can be written, by a point-transformation of coordinates, in the form of a warped Hamiltonian (\ref{Hext}). We will use the intrinsic characterization of such natural Hamiltonians making a simple generalization of the result given in \cite[Theorem~9]{CDRgen}

Let us consider a natural Hamiltonian
\begin{gather}\label{Hgen}
H=\frac 12 \tilde g^{ab}p_ap_b+\tilde V
\end{gather}
on a $(n+1)$-dimensional Riemannian manifold $\big(\tilde Q, \tilde{\mathbf g}\big)$ and let $X$ be a conformal Killing vector of $\tilde{\mathbf g}$, that is a vector field satisfying
 \begin{gather*}
[X,\tilde{\mathbf g}]=\mathcal L_X \tilde{\mathbf g}= \phi \tilde{\mathbf g},
 \end{gather*}
where $\phi$ is a function on $\tilde Q$, called {\em conformal factor}, and $[\cdot,\cdot]$ denotes the Schouten--Nijenhuis bracket. We define $X^\flat$ as the 1-form obtained by lowering the indices of the vector field $X$ with the metric tensor $ \tilde{\mathbf g}$.

\begin{teo}If on $\tilde Q$ there exists a conformal Killing vector field $X$ with conformal factor $\phi$ such that
\begin{gather}
{\rm d}X^\flat \wedge X^\flat=0, \label{Xnorm}\\
{\rm d}\phi\wedge X^\flat=0, \label{phi_u}\\
{\rm d} \|X \| \wedge X^\flat=0, \label{Xconst}\\
{\rm d}\big(X\big(\tilde V\big)+\phi\tilde V\big)\wedge X^\flat=0, \label{Vok}
\end{gather}
then, there exist on $\tilde Q$ coordinates $\big(u,q^i\big)$ such that
$\partial_u$ coincides up to a rescaling with $X$ and the natural Hamiltonian~\eqref{Hgen} has the form~\eqref{Hext}. Moreover, if
\begin{gather}
\tilde R(X)=aX, \qquad a\in \mathbb{R}, \label{Xeig}
\end{gather}
where $\tilde R$ is the Ricci tensor of the Riemannian manifold $\big(\tilde Q, \tilde{\mathbf g}\big)$,
then the function~$\alpha(u)$ in~\eqref{Hext} satisfies the condition $\alpha=-\gamma'$ with
$\gamma$ solution of~\eqref{eqgam}.
\end{teo}
\begin{proof}Following \cite{CDRgen}, conditions (\ref{Xnorm}), (\ref{phi_u}), (\ref{Xconst}) imply that the metric $\tilde{g}^{ab}$ is in the warped form (\ref{Hext}) in suitable canonical coordinates $(u,q^i)$ such that in these coordinates $X$ is of the form \begin{gather*}X=F(u)\partial_u \end{gather*} and
 \begin{gather*}\tilde{g}^{ij}\big(u,q^h\big)=F^{-2}g^{ij}\big(q^h\big),\qquad \phi=2\partial_uF. \end{gather*} Condition (\ref{Vok}) means that
$X\big(\tilde{V}\big)+\phi \tilde{V}$ is constant on the leaves orthogonal to $X$; in other words,
$X\big(\tilde{V}\big)+\phi \tilde{V}$ is a function of $u$ only.
Being $X=F(u)\partial_u$, $\phi=2 F'$ and $X\big(\tilde{V}\big)=F(u)\partial_u \tilde{V}$, last condition reads as $F(u)\partial^2_{ui}\tilde{V}+2F'\partial_i \tilde{V}=0$ for all $i$, that is
$\partial^2_{ui}\tilde{V}-\partial_u \ln F^2 \partial_i \tilde{V}=0$.
Hence, $\partial_u\big(\ln \big(F^2\partial_i \tilde{V}\big)\big)=0$ and therefore
$ F^2\partial_i \tilde V$ is independent of~$u$, that is $\partial_i \tilde{V}\big(u,q^h\big)=\partial_i V\big(q^h\big)F^{-2}$. It follows that
 \begin{gather*}\tilde{V}\big(u,q^h\big)= V\big(q^h\big)F^{-2} +f(u)=\alpha(u)V\big(q^h\big) +f(u) \end{gather*}
and thus it has the form of the potential in~(\ref{Hext}) where $V\big(q^h\big)$ is the potential of $L$.
The second part of the statement is proved in~\cite{CDRgen}.
\end{proof}

\begin{rmk} Conditions (\ref{Xnorm}), (\ref{phi_u}), (\ref{Xconst}) are the already known intrinsic characterization of a~warped metric, while condition (\ref{Vok}) intrinsically assures that the scalar part of the Hamiltonian is compatible with the warped structure. It is remarkable the fact that an intrinsic condition~-- equation~\ref{Xeig}~-- characterizes also the subset of the warped metrics which are involved in extended Hamiltonians. The extended potentials with $\Omega=0$ are intrinsically characterised by $X\big(\tilde{V}\big)+\phi\tilde{V}=0$.
\end{rmk}

We apply the above statement to the Hamiltonian~(\ref{H}) defined on Minkowski plane. First, we remark that on $\mathbb M^2$ the Ricci tensor is zero, therefore, the condition (\ref{Xeig}) is verified for any vector $X$ with $a=0$. A conformal Killing vector satisfying the remaining conditions is
\begin{gather*}
X=C\big(q^1 \partial_1+q^2 \partial_2\big), \qquad \phi=2C,
\end{gather*}
where $C$ is any non zero real constant.

The vector $X$ is the gradient of
the function $2Cq^1q^2$. Performing the coordinate transformation
\begin{gather*}
q^1=\frac 1{\sqrt{2}}\big(x^1+x^2\big), \qquad q^2=\frac 1{\sqrt{2}}\big(x^1-x^2\big),
\end{gather*}
we see that $\big(x^i\big)$ are the standard pseudo-Cartesian coordinates such that the geodesic term of the Hamiltonian $H$ becomes
 \begin{gather*}
\frac 12\big(p_1^2-p_2^2\big).
 \end{gather*}
In these coordinates, we have
 \begin{gather*}
2q^1q^2=\big(x^1\big)^2-\big(x^2\big)^2,
 \end{gather*}
thus, $2q^1q^2=C$ is the equation verified by points having fixed distance from the origin in the pseudo-polar coordinates of~$\mathbb M^2$. Hence, a coordinate system $(u,\chi)$ such that $X$ is parallel to~$\partial_u$ is the pseudo-polar one, with $u$ coinciding with the hyperbolic radius of the hyperbolic circles. Indeed, by applying the transformations
\begin{alignat*}{3}
& x^1=u\cosh \chi, \qquad && x^2=u \sinh \chi,& \\
& u=\sqrt{\big(x^1\big)^2-\big(x^2\big)^2}, \qquad && \chi= \tanh^{-1} \frac {x^2}{x^1},& \\
& q^1=\frac 1{\sqrt{2}} u {\rm e}^\chi, \qquad && q^2=\frac 1{\sqrt{2}} u{\rm e}^{-\chi},&
\end{alignat*}
where $\chi$ is the pseudo-angle, we get after simple computations that~(\ref{H}) becomes
\begin{gather*}
H=\frac 12 \left( p_u^2-\frac 1{u^2}p_\chi^2\right) -\frac 1{u^2}\left(2\alpha \big({\rm e}^{-2\chi}\big)^{2(k+1)}+\beta\big({\rm e}^{-2\chi}\big)^{k+1} \right).
\end{gather*}
Finally, we apply a further coordinate transformation, by setting
\begin{gather*}
\psi=(k+1)\chi,
\end{gather*}
and we get
\begin{gather} \label{HH}
H=\frac 12 p_u^2-\frac {(k+1)^2}{u^2}\left( \frac 12p_\psi^2 +\tilde \alpha {\rm e}^{-4\psi}+\tilde \beta {\rm e}^{-2\psi} \right),
\end{gather}
with $\tilde \alpha=\frac 2{(k+1)^2}\alpha $, $\tilde \beta=\frac {\beta}{(k+1)^2}$, which is an extended Hamiltonian on $\mathbb M^2$ with
 \begin{gather*}
(k+1)^2=-\frac{m^2}{n^2c}
 \end{gather*}
 of the one-dimensional Hamiltonian
\begin{gather}\label{26b}
L=\frac 12p_\psi^2 +V,
\end{gather}
where
\begin{gather*}
 V=\tilde \alpha {\rm e}^{-4\psi}+\tilde \beta {\rm e}^{-2\psi},
\end{gather*}
and $\gamma=\frac{1}{cu}$ satisfying (\ref{eqgam}) for $c_0=0$, $C=0$, hence with $\kappa=0$, and $c$ satisfying the condition
 \begin{gather*}
c=-\frac{m^2}{n^2(k+1)^2},
 \end{gather*}
i.e., $c=-\eta^2$ with $\eta\in\mathbb Q$.

As a second step, we have to find a non-trivial solution $G(p_\psi, \psi)$ of the equation
\begin{gather}\label{eq}
X_L^2G=-2cLG=2\eta^2LG.
\end{gather}
If we assume
 \begin{gather*}
G=g(\psi)p_\psi,
 \end{gather*}
by expanding equation (\ref{eq}) we obtain
\begin{gather*}
\big(g''-\eta^2g \big)p_\psi^3+\big[ \tilde \beta\big( 6g'-\big(4+2\eta^2\big)g \big) {\rm e}^{-2\psi}+\tilde \alpha \big(12 g'-\big(16+2\eta^2\big)g \big){\rm e}^{-4\psi}\big]p_\psi=0.
\end{gather*}
The coefficients of $p_\psi^3$ and $p_\psi$ vanish identically if and only if the following conditions hold
\begin{gather*}
g=a_1{\rm e}^{-\eta \psi}+a_2 {\rm e}^{\eta \psi},\\
g=a_3\big(\tilde \beta {\rm e}^{2\psi}+2\tilde \alpha \big)^{\frac{\eta^2}{12}-\frac 13}{\rm e}^{(\frac{\eta^2}6+\frac 43)\psi},
\end{gather*}
where $a_i$ are arbitrary real constants. A common solution $g$ exists for
 \begin{gather*}
\eta=2, \qquad a_1=0, \qquad a_2=a_3.
 \end{gather*}
Indeed, we get the non-trivial solution of (\ref{eq})
 \begin{gather*}
G={\rm e}^{2\psi}p_\psi, \qquad c=-4,
 \end{gather*}
that corresponds to the extended Hamiltonian
\begin{gather*}
H=\frac 12 p_u^2-\frac {m^2}{4n^2u^2}L, \qquad c=-4.
\end{gather*}
Hence, by solving
\begin{gather}\label{mn}
\frac mn =2(k+1),
\end{gather}
(for any $k\in \mathbb Q$, the equation admits positive integer solutions $m$, $n$) we can write the Hamiltonian~(\ref{HH}) as an extended Hamiltonian.

As immediate corollaries of the fact that the Hamiltonian~(\ref{H}) is an extension of (\ref{26b}), we get additional information about~(\ref{H})
\begin{enumerate}\itemsep=0pt
\item[(i)] the function $L$ is a quadratic first-integral of $H$, in accordance with the results of~\cite{CS};
\item[(ii)] as a consequence of (i), the Hamilton--Jacobi equation associated with $H$ is separable in a St\"ackel coordinates system for any $k\in \mathbb R$. It is evident from (\ref{HH}) that this coordinate system is $(u,\psi)$ (or, equivalently, $(u,\chi)$), the pseudo-polar coordinates of $\mathbb M^2$. The separability of the Hamiltonian~(\ref{H}) was not noticed until now, while the separability of (\ref{H0}) with $\beta=0$, $k\in \mathbb R$ was found in~\cite{MPT} in the same pseudo-polar coordinates.
\end{enumerate}

\begin{rmk}\label{na}
In \cite{MPT} it is proved the superintegrability of other two Hamiltonians in $\mathbb M^2$,
 \begin{gather*}
H_1=2 p_1p_2+\frac{q_2^d}{\sqrt{q_1}}, \qquad d=p, \qquad \text{or} \qquad d=\frac{1-2p}2, \qquad p\in \mathbb N,
 \end{gather*}
with the quadratic in the momenta first integral
 \begin{gather*}
I_1=2p_1(q_2p_2-p_1q_1)+\frac{q_2^{d+1}}{\sqrt{q_1}},
 \end{gather*}
and
 \begin{gather*}
H_2= 2 p_1p_2+q_1q_2^d, \qquad d=\frac{1-p}p, \qquad \text{or} \qquad d=\frac {1+2p}{1-2p},\qquad p\in \mathbb N,
 \end{gather*}
with first integral
 \begin{gather*}
I_2=p_1^2+\frac{q_2^{d+1}}{d+1}.
 \end{gather*}
These Hamiltonians cannot be written as extensions by using point transformations of coordinates. Indeed, due to the form of the extended Hamiltonians, in dimension two it is necessary that the extended Hamiltonian is diagonalized in the coordinates adapted to the warped product characterizing the extension. Necessarily, also the first integral~$L$, quadratic in the momenta, must be diagonalized in these coordinates, since the Hamiltonian is St\"ackel separable. It is easy to check that
$I_1$ and $I_2$ are the unique non-trivial quadratic first integrals of the Hamiltonians~$H_1$ and~$H_2$ respectively, and it is impossible to find coordinates, real or complex, simultaneously diagonalizing $H_1$ and $I_1$, or $H_2$ and $I_2$. It is possible to see that the Killing tensors associated to~$I_1$ and~$I_2$ have an eigenvalue with algebraic multiplicity two, but with one-dimensional eigenspace (in the case of indefinite metric there exist symmetric tensors which do not admit a basis of eigenvectors).
Therefore, the Killing tensors cannot determine separable coordinates~\cite{ben97}.
\end{rmk}

\section{Generalizations}\label{section4}

Since the theory of extended Hamiltonians can be applied to Hamiltonian (\ref{H}), which is the extension of (\ref{26b}) as it is proved in Section~\ref{s3}, we can obtain from (\ref{H}) some new families of extended Hamiltonians, which are maximally superintegrable systems on $\mathbb{M}^2$ and other two-dimensional Riemannian manifolds of constant curvature.

(i) The first evident generalization is that, from~(\ref{mn}), we have that~(\ref{H}) is an extended Hamiltonian for any
 \begin{gather*}
k=\frac {m}{2n}-1,
 \end{gather*}
and, therefore, for any rational number $k\in \mathbb Q$ instead of $k\in \mathbb N$. In \cite{MPT} the superintegrability of~(\ref{HH}) with $\beta=0$ was proved for $k\in \mathbb Q$.

 From Section~\ref{ex}, it is straightforward to find another generalization of the Hamiltonian~(\ref{H}) just by adding to it the scalar
 \begin{gather*}
\Omega u^2=2\Omega q^1q^2 , \qquad \Omega \in \mathbb R,
 \end{gather*}
so that we get a family of superintegrable potentials on $\mathbb{M}_2$ depending on three real parameters $(\alpha, \beta, \Omega)$, in analogy with the Tremblay--Turbiner--Winternitz potential on $\mathbb{E}_2$.
Indeed, the general theory of extended Hamiltonians resumed in Section~\ref{ex} assures the existence of a constant of motion of the form~(\ref{ee2}) for these new Hamiltonians.

(ii) We can try to find more general scalar potentials $V(\psi)$ in $L=\frac 12 p_\psi^2+V$ allowing the existence of functions $G$ solutions of~(\ref{eq}), where $c=-\eta^2$ in order to stay on a Lorentzian manifold, without restrictions on $\eta \in \mathbb R$.

We start by looking for $G=G(\psi)$. By solving equation (\ref{eq}) for both $G$ and $V$ as in the previous section, we find
\begin{gather}
G=C_1{\rm e}^{\eta \psi}-C_2{\rm e}^{-\eta \psi},\label{V10}\\
V=C_3 \big(C_1{\rm e}^{\eta \psi}+C_2{\rm e}^{-\eta \psi}\big)^{-2},\label{V1}
\end{gather}
where $C_1$, $C_2$, $C_3$ are arbitrary parameters (with the obvious constraints $C_3\neq 0$, $C_1$ and $C_2$ not both zero). These general expressions of $G$ and $V$ were obtained already in~\cite{CDRfi}.
Instead, by looking for $G$ of the form $G=g(\psi)p_\psi$, we obtain
\begin{gather}
g=C_1{\rm e}^{\eta \psi}+C_2{\rm e}^{-\eta \psi},\label{V20}\\
V=(C_1{\rm e}^{\eta \psi}+C_2{\rm e}^{-\eta \psi})^{-2}\big(C_3+C_4\big(C_1{\rm e}^{\eta \psi}-C_2{\rm e}^{-\eta \psi}\big) \big) =\left(C_3+\frac{C_4}{\eta} g'\right)g^{-2},\label{V2}
\end{gather}
where $C_1$, $C_2$, $C_3$, $C_4$ are arbitrary parameters (with the obvious constraints $C_1$, $C_2$ not both zero, as well as $C_3$, $C_4$).

It is immediate to see that (\ref{V1}) is a particular case of (\ref{V2}), obtained for $C_4=0$, while the functions $G$ and $g$ (given in (\ref{V10}) and (\ref{V20}), respectively) are equal up to the sign of~$C_2$. By setting
\begin{gather*}\eta=2, \qquad C_1=1, \qquad C_2=0, \qquad C_3=\tilde \alpha,\qquad C_4=\tilde \beta,
\end{gather*}
in (\ref{V20}), (\ref{V2}), we exactly find the potential of the Hamiltonian (\ref{26b}), with the corresponding function~$G$ solution of~(\ref{eq}).
In general, the change of coordinates
 \begin{gather*}
q^1=\frac u{\sqrt{2}} {\rm e}^{\frac \eta k \psi},\qquad q^2=\frac u{\sqrt{2}} {\rm e}^{-\frac \eta k \psi},
 \end{gather*}
transforms the extension
 \begin{gather*}
H=\frac 12 p_u^2- k^2\frac 1{\eta^2 u^2}\left(\frac 12 p_\psi^2+V\right)+\eta^4 \Omega u^2,\qquad k \in \mathbb Q,
 \end{gather*}
into
 \begin{gather*}
H=p_1p_2-\frac {4\tilde k^2}{2 \eta^2 q^1 q^2}\big(C_1z^{\tilde k}+C_2z^{-\tilde k}\big)^{-2}\big(C_3+C_4\big(C_1z^{\tilde k}-C_2z^{-\tilde k}\big) \big)+2\eta^4\Omega q^1 q^2,
 \end{gather*}
with
 \begin{gather*} z=\frac {q^1}{q^2}, \qquad \tilde k=\frac k2. \end{gather*}
It is now evident that, since $\eta$ is now an inessential multiplicative parameter of the potential, the extensions with a general $\eta \in \mathbb R$ are Hamiltonians essentially equivalent to those obtained for $\eta=2$.

We can write the functions (\ref{V20}) and (\ref{V2}) in a simpler form, focusing on essential parameters (with $\eta$ fixed equal to~2).
According to the sign of $C_1C_2$ we can write~(\ref{V20}) as
\begin{gather*}
g= a_0 {\rm e}^{ 2 \epsilon \psi}, \qquad \epsilon=\pm 1, \qquad V= \alpha {\rm e}^{- 2 \epsilon \psi} +\beta {\rm e}^{ -4\epsilon \psi}, \qquad C_1C_2=0,
\end{gather*}
where $a_0\neq 0$ is $C_1$ or $C_2$, $\epsilon=1$ if $C_2=0$ and $\epsilon=-1$ if $C_1=0$, $\alpha=C_4\epsilon a_0^{-1}$, $\beta=C_3a_0^{-2}$,
\begin{gather}
\label{Vcosh}
g=A\cosh (2\psi+\psi_0), \qquad V= \dfrac{\alpha+\beta\sinh(2\psi+\psi_0)}{\cosh^2(2\psi+\psi_0)} ,\qquad C_1C_2>0, \\
\label{VTTW}
 g=A\sinh (2\psi+\psi_0) ,\qquad V= \dfrac{\alpha+\beta\cosh(2\psi+\psi_0)}{\sinh^2(2\psi+\psi_0)} , \qquad C_1C_2<0,
\end{gather}
where
$C_1=A{\rm e}^{\psi_0}$, $C_2=\frac{|C_1C_2|}{C_1C_2}A{\rm e}^{-\psi_0}$,
$\psi_0=\tfrac 12 \ln|C_1/C_2|$, $\alpha=C_3/A^2$, $\beta=C_4/A$.

The form~(\ref{VTTW}) appeared already in~\cite{TTWcdr} and it is mapped
to (\ref{Vcosh}) by an imaginary translation of~${\rm i}\pi/2$. Moreover,
the transformation
 \begin{gather*}
q^1=\frac u{\sqrt{2}} {\rm e}^{\frac 2 k \psi},\qquad q^2=\frac u{\sqrt{2}} {\rm e}^{-\frac 2 k \psi},
 \end{gather*}
maps the extended Hamiltonians
 \begin{gather*}
\frac 12 p_u^2 - k^2\frac 1{4 u^2}\left(\frac 12 p_\psi^2+V\right)+16 \Omega u^2,\qquad k \in \mathbb Q,
 \end{gather*}
into
 \begin{gather*}
H=p_1p_2-\frac {\tilde k^2}{2 q^1 q^2}\big(C_1z^{\tilde k}+C_2z^{-\tilde k}\big)^{-2}\big(C_3+C_4\big(C_1z^{\tilde k}-C_2z^{-\tilde k}\big) \big)+32\Omega q^1 q^2.
 \end{gather*}

\begin{rmk} The potentials $V$ described in~(\ref{V2}) are always algebraic functions of $q^1$, $q^2$, provided $\eta \in \mathbb Q$ (which in this case is absorbed by the rational~$k$).
\end{rmk}

The generalized extended Hamiltonian obtained from (\ref{H}) is therefore
\begin{gather}\label{ge}
\tilde H=\frac 12 p_u^2-\frac{m^2}{\eta^2 n^2 u^2}\left(\frac 12 p_\psi^2+V \right)+\eta^4\Omega u^2,
\end{gather}
where $V$ is (\ref{V2}) and $\frac mn=\eta(k+1)$, $c=-\eta^2$.

\begin{rmk} The Hamiltonian (\ref{ge}) becomes (\ref{HH}) when $C_2=0$, $C_1=1$, $C_3=\tilde \alpha$, $C_4=\tilde \beta$. When $C_2=1$, $C_1=0$, the potentials correspond up to the exchange between~$q^1$ and $q^2$. For~$C_1$ and $C_2$ both $\neq 0$, we obtain new superintegrable potentials. Of the four parameters $C_i$ only two are essential, the remaining two being reducible to a single multiplicative factor.
\end{rmk}

One could search for other superintegrable potentials $V$ by assuming different expressions for $G=G(\psi, p_\psi)$, but the resulting differential equations are much less easy to solve than for $G$ constant or linear in the momentum~$p_\psi$.

(iii) From the Hamiltonian (\ref{ge}) we can obtain another family of superintegrable Hamiltonians depending on a rational parameter~$k$ by applying to it the coupling-constant metamorphosis (CCM) as described in \cite{CDRpw, Hi, KMccm, PWccm}.

The CCM transforms integrable or superintegrable systems in new integrable or superintegrable ones, by mapping first integrals in first integrals. It is characterized by the following theorem

{\it Let us consider a Hamiltonian $H =\hat H - \tilde EU$ in canonical coordinates $\big(x^i,p_i\big)$, where $\hat H\big(x^i,p_i\big)$ is independent of the arbitrary parameter $\tilde E$ and $U\big(x^i\big)$, with an integral of the motion $K$ $($depending on~$\tilde E)$. If we define the CCM of $H$ and $K$ as $ H' = U^{-1}\big( \hat H -E\big)$ and $ K' = K|_{\tilde E=\tilde{H}}$ then~$ K'$ is an integral of the motion for~$H'$.}

In our case, the configuration manifold of the resulting system is again~$\mathbb M^2$. Indeed, from~(\ref{ge}), by applying the CCM based on $\tilde E=-\eta^4\Omega$, and recalling that $k+1=\frac mn$, we get
\begin{gather*}
H'= \frac 1{2u^2} p_u^2-\frac{m^2}{\eta^2n^2u^4}\left(\frac 12 p_\psi^2+V \right)-\frac{E}{u^2},
\end{gather*}
where $V$ is given by (\ref{V2}). By performing the rescaling $u^2=2v$, we obtain
\begin{gather}\label{H2}
H'= \frac{1}{2}p_v^2-\frac{m^2}{4\eta^2n^2v^2}\left(\frac 12 p_\psi^2+V \right)-\frac E{2v},
\end{gather}
or, in coordinates
 \begin{gather*}
q^1=\frac v{\sqrt{2}} {\rm e}^{\frac \eta k \psi},\qquad q^2=\frac v{\sqrt{2}} {\rm e}^{-\frac \eta k \psi},
\\
H'=p_1p_2-\frac {4\tilde k^2}{2 \eta^2 q^1 q^2}\big(C_1z^{\tilde k}+C_2z^{-\tilde k}\big)^{-2}\big(C_3+C_4\big(C_1z^{\tilde k}-C_2z^{-\tilde k}\big)\big)-\frac E{2\sqrt{2q^1 q^2}},
 \end{gather*}
with
 \begin{gather*} z=\frac {q^1}{q^2}, \qquad \tilde k=\frac m{4n}. \end{gather*}

The configuration manifold of $H'$, is again $\mathbb~M^2$. The form of the characteristic first integral of~(\ref{H2}) can be found in~\cite{CDRpw}.

(iv) The choice $c=-\eta^2$, for $\eta =2$, and $C=\kappa=0$ in Section~\ref{s3} has been made in order to obtain the Hamiltonian (\ref{H}) in $\mathbb M^2$ as an extension of~(\ref{26b}). However, from Section~\ref{section1} we know that for different values of $c$ and $C$ we can obtain extensions of the same~$L$ of~(\ref{26b}) which are defined on
manifolds different from~$\mathbb M^2$, but which still belong to the same family of extensions of~$L$ and, therefore, which can be seen as generalizations of the Hamiltonian (\ref{H}) in these new manifolds.
A first possible generalization is the following: we consider the expressions~(\ref{V20}),~(\ref{V2}) for~$G$ and~$V$ and we allow that $\eta$ takes imaginary values. Such values of $\eta$ imply the appearance of trigonometric functions in~$G$ and~$V$, which remain real provided $C_1=\bar C_2$ and $C_4$ are imaginary. In this way we see that the extension
\begin{gather}\label{h3}
H_0=\frac 12 p_u^2+\frac {m^2}{|\eta|^2n^2u^2}L+\eta^4u^2\Omega
\end{gather}
is now defined in the Euclidean plane, since $c=|\eta|^2>0$. It follows that $H_0$ is nothing but the Tremblay--Turbiner--Winternitz Hamiltonian discussed in \cite{TTWcdr}.

 Moreover, we can consider values of $\kappa$ different from 0, according to the freedom of choice of the parameter $C$ in (\ref{eqgam}). Consequently, for $\eta$ real or imaginary (i.e., for $c<0$ or $c>0$), we have the possible extensions of $L$
 \begin{gather}\label{h1}
H_1=\frac 12 p_u^2+\frac {(k+1)^2c}{\sin^2 cu}L+\Omega \tan^2 cu,
\qquad \kappa=1,
\\
\label{h2}
H_2=\frac 12 p_u^2+\frac {(k+1)^2c}{\sinh^2 cu}L+\Omega \tanh^2 cu,
\qquad \kappa=-1,
\end{gather}
 corresponding, for $c>0$, to Hamiltonians on the sphere $\mathbb S^2$ and on the pseudosphere $\mathbb H^2$, respectively, and, for $c<0$, to Hamiltonians defined on the de~Sitter space ${\rm dS}^2$ and on the anti-de~Sitter space ${\rm AdS}^2$, respectively.
We recall that, in this case, the value of~$c$ is related to the value of the curvature of the configuration manifold of~$H_1$,~$H_2$.

In this way, we obtained some kind of generalizations of the Hamiltonian~(\ref{H}) in manifolds of constant positive and negative curvature, on Riemannian and pseudo-Riemannian manifolds, since they are all obtained as extensions of the same one-dimensional Hamiltonian.

 Finally, we consider the case $c=0$, which cannot be deduced from (\ref{V20}), (\ref{V2}), since we assumed there that $c_0=0$ in (\ref{e1}).
For the case $c=0$, $c_0\neq0$, we have from \cite{TTWcdr} either
 \begin{gather*}g=(a_1 \psi +a_2),\qquad V=\frac {c_0}{4a_1^2} g^2+c_1 g+c_2, \end{gather*}
or
 \begin{gather*} g=a_2,\qquad V=c_0 \psi^2 +c_1 \psi +c_2. \end{gather*}
Therefore, we have that the function~$L$~(\ref{26b}) is not included in this case.

 The Hamiltonians (\ref{h1}), (\ref{h2}), defined on conformally flat Riemannian manifolds, can be again transformed into new superintegrable systems by applying to them the CCM.

The results of above are summarized in the following statement

\begin{teo}\label{gen} The Hamiltonians \eqref{ge}, \eqref{H2}, \eqref{h3}, \eqref{h1}, \eqref{h2} are superintegrable with a~first integral polynomial in the momenta computed through the procedure described in Section~{\rm \ref{ex}}.
\end{teo}

\begin{rmk} The Hamiltonian
 \begin{gather*}
\tilde H=p_1p_2-\alpha q_2^{2k+1}q_1^{-2k-3}-\frac \beta 2 q^k_2q_1^{-k-2} +2 \Omega q^1q^2, \qquad k\in \mathbb{R},
 \end{gather*}
as well as any Hamiltonian if the form
 \begin{gather*}
H=p_1p_2+\frac 1{q^1q^2}F\left(\frac {q^1}{q^2}\right)+G\big(q^1q^2\big),
 \end{gather*}
remain separable in pseudo-polar coordinates for any $k\in \mathbb R$, and therefore it is Liouville integrable.
\end{rmk}

\begin{rmk} All the Hamiltonians of Theorem \ref{gen} correspond essentially to the Tremblay--Turbiner--Winternitz system, its generalizations on constant-curvature manifolds and its CCM transformed Hamiltonians. This becomes evident by comparing the Hamiltonians of above with those given in \cite{TTWcdr, KMK, MPW} and references therein.
\end{rmk}

\section{Quantization}\label{section5}

We consider the Hamiltonian (\ref{ge}) with $V$ given by (\ref{V2}), thus including the Hamiltonian (\ref{H}). We indicate how to quantize it together with its characteristic first integral corresponding to any rational parameter, proving the quantum superintegrability of this two-dimensional system with a potential involving three arbitrary parameters. The quantization is made possible by applying the method described in~\cite{CRsl}, where the Kuru--Negro quantization introduced in~\cite{KN}, based on factorization in shift and ladder operators, is adapted to extended Hamiltonians on flat manifolds.
 Indeed, in order to apply the method for $H$ defined on a two-dimensional manifold, it is necessary to find functions of the form
 \begin{gather*}
 F^{\pm}=\pm F'(\psi)p_\psi+F(\psi)f+\frac 1fc_1, \qquad f=\sqrt{2\big(\eta^2L+c_0\big)},
 \end{gather*}
 solutions of
 \begin{gather}\label{le}
 X_L^2F^\pm = f F^\pm.
 \end{gather}
 The functions $F^\pm$ are known as {\em ladder functions} of $L$ and they are determined by solutions $F(\psi)$ of the system
 \begin{gather*}
 F''-\eta^2F=0,\\
 V'F'+2\eta^2 VF+c_1=0,
 \end{gather*}
equivalent to (\ref{le}) and corresponding to~(20) and~(21) of~\cite{CRsl}, for a suitable constant~$c_1$. We have in our case
 \begin{gather*}
 F=g', \qquad c_1=-C_4\eta^3,
 \end{gather*}
 where $g$ is given in~(\ref{V20}).
 The construction of the quantum Hamiltonian and of the characteristic symmetry operator is then explicitely described in \cite[Theorem~10]{CRsl}. To the function $L$ we associate the quantum operator
 \begin{gather*}
 \hat L=-\frac {\hbar^2} 2\partial^2_{\psi \psi}+V,
 \end{gather*}
 which can be extended to the quantum Hamiltonian
 \begin{gather*}
 \hat H=-\frac {\hbar^2} 2\left( \partial ^2_{uu}+\frac 1u\partial_u -\frac{k^2}{\eta^2u^2} \partial^2_{\psi \psi} \right)-\frac{k^2}{\eta^2u^2}V+\frac{\Omega \eta^4}2u^2, \qquad k=\frac mn, \qquad m,n \in \mathbb N\setminus \{0\},
 \end{gather*}
 which commutes with $\hat L$. The operator $\tilde H$ coincides with the Laplace--Beltrami operator of the metric
\begin{gather*}
 {\rm d}s^2={\rm d}u^2-\frac{\eta^2u^2}{k^2}{\rm d}\psi^2,
\end{gather*}
 which, for $\eta^2=1$, is the metric of the Minkowski plane in pseudo-polar coordinates $(u,k\psi)$, plus a scalar term.
 Then, the operator
 \begin{gather*}
 \hat{X}^k_{\epsilon,E}=\big(\hat{G}^{+}_\epsilon\big)^{2n} \circ
 \big(\hat{A}^{1,1}_{k\epsilon}\big)^{2m}\circ \big(\hat{D}^{+}_E\big)^m,
 \end{gather*}
 is such that
 \begin{gather*}\label{hHhX}
 \hat{H}\big(\hat{X}^k_{\epsilon, E} f_E\big)= E\hat{X}^k_{\epsilon, E} (f_E),
 \end{gather*}
 for all multiplicatively separated eigenfunctions $f_E=\phi_E^{M}(u)\chi_\lambda(\psi)$ of~$\hat{H}$ such that $\chi_\lambda$ is an eigenfunction of $\hat{L}$ with eigenvalue~$\lambda$, and
 \begin{gather*}M=k^2\lambda,\qquad \epsilon
 =\sqrt{| \lambda |}, \qquad
 \eta^2\lambda\leq 0, \end{gather*}
 where $\phi^M_E(u)$ is (not restrictively) any eigenfunction of the ``radial'' operator
 \begin{gather*}\hat H^M=-\frac {\hbar^2} 2\left( \partial ^2_{uu}+\frac 1u\partial_u \right)-\frac M{\eta^2u^2}, \end{gather*}
 corresponding to the eigenvalue $E$, and where
 \begin{gather*}
 \hat A^{1,1}_{k\epsilon}=\partial_u-\frac 1\hbar \left(\eta^2\sqrt{\Omega}u+\frac{k\epsilon \sqrt{2}}{2\eta u} \right),\\
 \hat D^\pm_E=\frac \hbar{\sqrt{2}}u\partial_u+\frac \hbar{\sqrt{2}}\pm\left(-\frac E{2\eta^2\sqrt{\Omega}}+\sqrt{\frac{\Omega}2}\eta^2 u^2 \right),\\
 \hat G^{\pm}_\epsilon=\eta^2 g\partial_\psi\pm\epsilon \eta \frac {\sqrt{2}}\hbar g'+\frac{2 \eta^2 C_4}{\hbar\big(\hbar \eta\pm 2\sqrt{2}\epsilon\big)}.
 \end{gather*}

The operator $\hat X^k_{\epsilon,E}$ is therefore a symmetry of $\hat H$, called ``warped symmetry'' in~\cite{CRsl}, corresponding to the characteristic first integral of the classical extended Hamiltonian~$H$.

\section{Conclusions}
In this article we show that the theory of extended Hamiltonians allows not only to prove the superintegrability of the Hamiltonian~(\ref{H}) for any non-zero rational $k$, but also gives immediate generalizations of~$H$, both on flat and curved manifolds, which are still superintegrable Hamiltonians and include the superposition of a harmonic oscillator term.

Moreover, the application of the coupling-constant-metamorphosis produces further superintegrable Hamiltonians, when adapted to the structure of extended Hamiltonians.

The possible generalizations are not exhausted by the examples given in the present paper. For example, the methods employed in~\cite{CDRsuext} allow to extend $n$-dimensional superintegrable Hamiltonians into $(n+1)$-dimensional ones, again superintegrable. The application of those methods to the present case has not been considered yet.

The procedure of Laplace--Beltrami quantization of extended Hamiltonians and their associated first integrals of high-degree developed in \cite{CRsl} on flat manifolds is applied to the present case, so the quantum superintegrability of (\ref{H}) and its generalisations on flat manifolds is proved. The problem of the quantum superintegrability of (\ref{H}) was left open in~\cite{CS}. The quantization of extended Hamiltonians on curved manifolds, such as (\ref{h1}) and (\ref{h2}), is still an open problem.

The search for extended Hamiltonian structures is limited for the moment to point-coordinate transformations. A possible future direction of research is the search for extended Hamiltonian structures under more general canonical transformations, for example concerning the systems described in Remark~\ref{na}.

\pdfbookmark[1]{References}{ref}
\LastPageEnding

\end{document}